\documentclass[conference]{IEEEtran}
\IEEEoverridecommandlockouts
\usepackage{cite}
\usepackage{algorithm}
\usepackage{amsmath,amssymb,amsfonts}
\usepackage{algorithmic}
\usepackage{graphicx}
\usepackage{stfloats}
\usepackage{textcomp}
\usepackage{array}
\usepackage{xcolor}
\usepackage{amssymb, amsthm}
\usepackage{subcaption}
\usepackage[bottom=0.98in, left = 0.6in, right=0.6in]{geometry}  

\newcommand{\cM}{\mathcal{M}}

\newcommand{\cN}{\mathcal{N}}

\newcommand{\balpha}{\boldsymbol \alpha}

\newcommand{\bo}{\boldsymbol o}

\newcommand{\bx}{\boldsymbol x}

\newcommand{\bP}{\boldsymbol P}

\newtheorem{myth}{Theorem}
\newtheorem{myprop}[myth]{\bf Proposition}

\newtheorem{myle}[myth]{\bf Lemma}

\def\BibTeX{{\rm B\kern-.05em{\sc i\kern-.025em b}\kern-.08em
    T\kern-.1667em\lower.7ex\hbox{E}\kern-.125emX}}

\begin{document}

\title{Power Minimization for NOMA-assisted Pinching Antenna Systems With Multiple Waveguides}
\author{Yaru Fu, Fuchao He, Zheng Shi, and Haijun Zhang,~{\IEEEmembership{Fellow,  IEEE}}

\thanks{
This work was supported in part by the grant from the Research Grants Council (RGC) of the
Hong Kong Special Administrative Region, China, under Reference No.
UGC/FDS16/E02/22, in part by the Hong Kong Research Matching Grant (RMG) in the Central Pot under Project No. CP/2022/2.1, in part by the Team-based Research Fund under Project No. TBRF/2024/1.10, and in part by the National Natural Science Foundation
of China under Grant 62171200.

Y. Fu is with the School of Science and Technology,  Hong Kong Metropolitan University, Hong Kong SAR, 999077, China (e-mail: yfu@hkmu.edu.hk).

F. He and Z. Shi are with the School of Intelligent Systems Science and Engineering, Jinan University, Zhuhai 519070, China (e-mail: hefuchao517@gmail.com;~zhengshi@jnu.edu.cn).

Haijun Zhang is with the National School of Elite
Engineering, Beijing Engineering and Technology Research Center for Convergence Networks and Ubiquitous Services, University of Science and
Technology Beijing, Beijing, 100083, China (e-mail: haijunzhang@ieee.org).
 }}

\maketitle

\begin{abstract}
The integration of pinching antenna systems with non-orthogonal multiple access (NOMA) has emerged as a promising technique for future 6G applications. This paper is the first to investigate power minimization for NOMA-assisted pinching antenna systems utilizing multiple dielectric waveguides. We formulate a total power minimization problem constrained by each user's minimum data requirements, addressing a classical challenge. To efficiently solve the non-convex optimization problem, we propose an iterative algorithm. Furthermore, we demonstrate that the interference function of this algorithm is standard, ensuring convergence to a unique fixed point. Numerical simulations validate that our developed algorithm converges within a few steps and significantly outperforms benchmark strategies across various data rate requirements. The results also indicate that the minimum transmit power, as a function of the interval between the waveguides, exhibits an approximately oscillatory decay with a negative trend.
\end{abstract}

\smallskip

\begin{IEEEkeywords}
Minimum data rate, non-orthogonal multiple access (NOMA), pinching antenna, and power minimization.
\end{IEEEkeywords}

\section{Introduction}
The forthcoming sixth-generation wireless cellular networks (6G) aim to deliver unprecedented data rates and seamless connectivity across diverse applications \cite{DT_self2}. To meet the stringent quality of service demands of 6G, advanced technologies are being explored to enhance system capacity and reliability. Among these, non-orthogonal multiple access (NOMA) has emerged as a key enabler due to its capability of improving spectral efficiency by allowing multiple users to share the same resource block through power-domain multiplexing \cite{Yuanwei_NOMA}. Unlike traditional orthogonal schemes, NOMA maximizes spectrum utilization, making it a promising solution for high-capacity networks. Meanwhile, technologies such as intelligent reflecting surfaces (IRS), fluid antennas, and movable antennas offer dynamic channel adaptation to mitigate interference and improve signal quality \cite{movable,6G_2025}.

However, maintaining reliable line-of-sight (LoS) links remains a significant challenge for 6G networks operating in high-frequency millimeter wave and terahertz bands. These bands, while supporting massive data rates, are highly susceptible to blockages from buildings, trees, and even moving objects, making uninterrupted LoS essential. Movable and fluid antennas offer limited repositioning capabilities, often restricted to a few wavelengths, which proves insufficient against large or immovable obstacles \cite{ding}. IRS presents an alternative by reflecting signals to form virtual LoS paths; however, this approach doubles the transmission distance, leading to ``double path loss" and significantly weakening signal strength \cite{hong_IRS}. Moreover, IRS requires precise placement and fine-tuned control, increasing deployment complexity and cost. These limitations highlight the need for a more flexible and resilient solution to ensure robust LoS connectivity in 6G networks.

To address these challenges, pinching antennas have been proposed \cite{PA_docomo}, which utilize dielectric waveguides and specialized materials to pinch the waveguides, enabling controlled radiation of signals into the surrounding environment. A key advantage of pinching antennas is their ability to slide along the waveguide, dynamically adjusting radiation points to create on-demand LoS links. This flexibility surpasses traditional fixed-location antennas, offering a more adaptive and efficient solution for overcoming LoS blockages.
Integrating NOMA with pinching antennas combines their respective advantages to enhance system performance further  \cite{PA_ding,PA_yunhui}. NOMA optimally allocates power among multiple users, leveraging the strong LoS links created by pinching antennas, improving capacity, fairness, and spectral efficiency. By integrating NOMA's resource-sharing capability with pinching antennas' adaptable LoS provisioning, a robust framework for next-generation wireless networks can be realized, paving the way for high-performance 6G communications.

Some initial studies have explored the integration of NOMA with pinching antenna systems \cite{PANOMA_rate,NOMA_PA_ding}. More specifically, in \cite{PANOMA_rate}, the authors aimed to maximize the sum rate of pinching antenna-assisted NOMA systems with multiple dielectric waveguides. They proposed an alternating optimization algorithm to obtain a suboptimal solution in polynomial time. In \cite{NOMA_PA_ding}, the optimization of pinching antenna placement and activation was investigated to maximize the throughput of a NOMA system, considering a single waveguide setup. Although the effectiveness of pinching antenna-assisted NOMA systems has been validated, existing studies have focused solely on rate maximization. The problem of minimizing total power consumption while satisfying minimum data rate requirements, a fundamental challenge in wireless cellular networks, has yet to be fully addressed. This gap motivates our work. For simplicity, our main contributions are outlined below:
\begin{itemize}
    \item We consider a NOMA-assisted pinching antenna system with multiple dielectric waveguides, where each waveguide serves multiple NOMA users. In this setting, a user's achievable rate is jointly determined by the power control methods of all other users associated with the same or different waveguides.
    \item We formulate a classical power control problem that aims to minimize the total power consumption of all users while considering the minimum data rate requirements. We propose an iterative power allocation approach that can achieve a unique fixed point within a few steps. Furthermore, we prove that the interference function of the developed method is standard, which guarantees the convergence performance of the solution.
    \item Finally, extensive numerical simulations are presented to validate the effectiveness and efficiency of our developed method. The results reveal that the minimum required transmission power, in relation to the interval between waveguides, exhibits an approximate oscillatory decay with a negative exponential envelope. A rigorous analysis is also provided to explain the underlying reasons for this behavior.
\end{itemize}
The remainder of this paper is organized as follows. In Section II, we elaborate on the details of the system model for our considered multiple waveguide-derived NOMA-pinching antenna system. The formulated minimization problem is presented in Section III. An iterative algorithm and its property analysis are provided in Section IV. Simulation results are discussed in Section V. Finally, we conclude the paper and outline future research directions in Section VI.
\section{System Model}
\begin{figure}[t]
    \centering
    \includegraphics[width=0.93\linewidth]{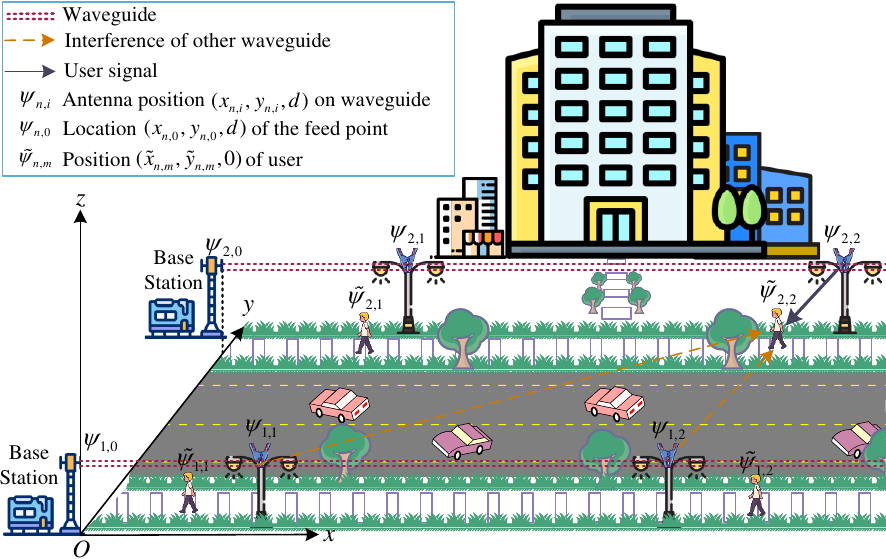}
    \caption{System model of the NOMA-assisted pinching antenna network with multiple waveguides.}
    \label{fig:system_model}
\end{figure}

As illustrated in Fig. \ref{fig:system_model}, we consider a pinching antenna system with $N$ waveguides, each equipped with $M$ pinching antennas to accommodate $M$ mobile users using NOMA. Define $\psi_{n,i}=(x_{n,i},y_{n,i},d)$ as the location of the $i$-th antenna on waveguide $n$, where $d$ depicts the hight of the waveguide. In addition, we use $\tilde \psi_{n,m}=(\tilde x_{n,m},\tilde y_{n,m},0)$ to represent the position of user $m$ that is served by antenna $m$ on waveguide $n$. 
The transmit power for waveguide $n$ is denoted by $P_n$.  Since NOMA is used, the transmit signal of waveguide $n$ can be expressed as follows:
\begin{equation}
\bx_n=\sum_{m=1}^{M}\sqrt{P_n\alpha_{n,m}}\bx_{n,m},
\end{equation}
where $\bx_{n,m}$ and $\alpha_{n,m}$ denote the signal and the power allocation coefficient for user $m$ served by waveguide $n$, respectively.

Let $h_{n,m}$ be the channel gain of the user served by antenna $m$ on waveguide $n$. Based on \cite{ding}, it is given below:
\begin{equation}\label{channel_gain}
    h_{n,m}=\sum_{i=1}^{M} \frac{\sqrt{\eta}e^{-j\frac{2\pi}{\lambda}\lVert \tilde \psi_{n,m}- \psi_{n,i} \rVert}}{\lVert\tilde \psi_{n,m}- \psi_{n,i}\rVert}e^{-j\theta_{n,i}},
\end{equation}
in which $\eta=\frac{c^2}{16\pi^2f^2_c}$, $c$ is the speed of light, $f_c$ is the carrier frequency, $\theta_{n,i}=2\pi\frac{\psi_{n,0}-\psi_{n,i}}{\lambda_g}$ depicts the phase shift experienced at the $i$-th antenna on waveguide $n$ and $\lVert \cdot \rVert$ expresses Euclidean norm. Thereof, $\psi_{n,0}$ and  $\lambda_g$ are the location of the feed point of waveguide $n$ and the waveguide wavelength in a dielectric waveguide \cite{ding}.
Moreover, we define $h_{n',n,i,m}$ as the channel gain between the $i$-th antenna on the waveguide $n'$ and the $m$-th user associated with the waveguide $n$. Mathematically, we have
\begin{equation}
    h_{n',n,i,m}=\frac{\sqrt{\eta}e^{-j2\pi(\frac{1}{\lambda}\lVert\tilde\psi_{n,m}-\psi_{n',i}\rVert+\frac{1}{\lambda_g}\lVert\psi_{n',0}-\psi_{n',i}\rVert)}}{\lVert\tilde\psi_{n,m}-\psi_{n',i}\rVert}.
\end{equation}

For NOMA users associated with the same waveguide, it is challenging to identify the strong and weak ones, as their signal strength is primarily affected by their locations, antenna positions, and the transmit power of the other users. To characterize this, we use the normalized signal-to-inter-waveguide-interference-plus-noise ratio (nSIINR) for user $m$ served by the waveguide $n$, denoted as $\text{nSIINR}_{n,m}$, as a metric \cite{Yaru_TWC}. With the foregoing definitions,  $\text{nSIINR}_{n,m}$ is given as follows:
\begin{equation}  \label{nsinr}
\text{nSIINR}_{n,m}=\frac{P_n\alpha_{n,m}}{\sum_{n'=1,n'\neq n}^{N}\sum_{i=1}^M|\bar h_{n',n,i,m}|^2P_{n'}\alpha_{n',i} +\bar\sigma_{n,m}^2},
\end{equation}
where
$$|\bar h_{n',n,i,m}|^2=\frac{|h_{n',n,i,m}|^2}{|h_{n,m}|^2},$$
and $$\bar\sigma_{n,m}^2=\frac{\sigma^2}{|h_{n,m}|^2},$$
in which $\sigma^2$ represents the additive white Gaussian noise power.
For the users served by the waveguide $n$, we rank their nSINRs in ascending order. Denote by $\bo_n=(o_1, o_2, \ldots, o_M)$, which satisfies $\text{nSIINR}_{n,o_i} < \text{nSIINR}_{n,o_j}$ if $i<j$. Following the downlink NOMA decoding principle, the interference plus noise ratio of user $o_m$ served by the waveguide $n$, denoted by $\text{SINR}_{n,o_m}$, can be expressed as follows:
\begin{equation} \label{sinr}
    \text{SINR}_{n,o_m}=\frac{P_n\alpha_{n,o_m}}{\sum_{j=m+1}^M P_n\alpha_{n,o_j}+\text{IIN}_{n,o_m}},
\end{equation}
where $\text{IIN}_{n,o_m}$ can be taken as the \textit{normalized} inter-waveguide interference plus noise experienced by user $o_m$, which is quoted below:
\begin{equation} \label{IN}
\text{IIN}_{n,o_m}=\sum_{n'=1,n'\neq n}^{N}\sum_{i=1}^M|\bar h_{n',n,i,o_m}|^2P_{n'}\alpha_{n',i} +\bar\sigma_{n,o_m}^2.
\end{equation}
With the analysis, the data rate of this user can be expressed as follows:
\begin{equation} \label{rate}
R_{n,o_m}=W\log_2(1+\text{SINR}_{n,o_m}),
\end{equation}
where $n\in\cN$, $m\in\cM$, and $\text{SINR}_{n,o_m}$ is given in \eqref{sinr}.
\section{Problem Formulation}
In this paper, our objective is to minimize the total transmit power of all users while satisfying the data rate requirement of each user. For simplicity, we define $\balpha_n=(\alpha_{n,1}, \alpha_{n,2}, \ldots, \alpha_{n,M})$ as the power allocation coefficients for users associated with the waveguide $n$. In addition, let $\balpha=(\balpha_1, \balpha_2, \ldots, \balpha_N)$ be the power control coefficients for all users.  We also denote $\bP=(P_1,P_2,\ldots,P_N)$ as the power allocation vector of all the waveguides. Furthermore, denote by $\bar R_{n,m}$ the minimum data rate requirement for user $m$ served by the waveguide $n$. Mathematically, we have
\begin{align} \label{obj}
	& \underset{\balpha,\bP}{\text{min}} \sum_{n=1}^NP_n  \\
	\text{s.t.}\nonumber~
	&\text{C1}:~\sum_{m=1}^M \alpha_{n,m}=1,~n\in\cN,~\nonumber\\	&\text{C2}:~0\leq\alpha_{n,m}\leq1,~n\in\cN,~m\in\cM,\nonumber\\
	&\text{C3}:~0 \leq P_n\leq \bar P_n, ~n\in\cN,\nonumber\\
	&\text{C4}:~R_{n,m} \geq \bar R_{n,m},~n\in\cN,~m\in\cM,\nonumber
\end{align}
where C1 indicates that, for each waveguide, the sum of the power coefficients of its served users should be 1. C2 ensures that the power allocation coefficient for each user lies between 0 and 1.   C3 limits the maximum transmission power of waveguide $n$ to $\bar P_n$. Finally, C4 ensures that each user achieves its minimum data rate. The formulated optimization problem \eqref{obj} is non-convex because $R_{n,m}$ is non-convex with respect to $\bP$. Therefore, obtaining its optimal solution is challenging. In the subsequent section, we will present our developed time-efficient power allocation algorithm with property assurance.

\section{Algorithm Design and Property Analysis} \label{alg}
In this section, we introduce a distributed power allocation scheme for NOMA-assisted pinching antenna systems.
This method is designed to achieve the fixed point of our formulated problem \eqref{obj}.
\subsection{Algorithm Design}
Before elaborating on the detailed algorithm design, we give the following lemma.
\begin{myle}
   The optimal solution to problem \eqref{obj} is achieved when the inequalities in C4 hold as equalities.
\end{myle}
The lemma can be proved using the method of contradiction, but the details are omitted here for brevity. We start with user $o_M$, who is associated with waveguide $n$, as an example and assume that the transmit power of all users associated with the other $N-1$ waveguides is given. According to \eqref{rate}, Lemma 1, and constraint C4, the minimum required transmit power for user $o_M$ can be straightforwardly obtained. This power depends solely on the transmit power of all other users with $N-1$ waveguides and is expressed as follows:
\begin{equation} \label{p_M}
p_{n,o_M}=(2^{\frac{\bar R_{n,o_M}}{W}}-1)\text{IIN}_{n,o_M},
\end{equation}
where $\text{IIN}_{n,o_M}$ is defined in \eqref{IN}.

Next, we focus on user $o_{M-1}$, who is served by waveguide $n$ as well. Based on \eqref{rate}, \eqref{p_M}, Lemma 1, and constraint C4, we can calculate the minimum required transmission power for this user, as shown below:
\begin{equation} \label{p_M-1}
p_{n,o_{M-1}}=(2^{\frac{\bar R_{n,o_{M-1}}}{W}}-1)(p_{n,o_M} +\text{IIN}_{n,o_{M-1}}),
\end{equation}
in which $\text{IIN}_{n,o_{M-1}}$ and $p_{n,o_M}$ are obtained in \eqref{IN} and \eqref{p_M}, respectively.

By following similar procedures, we can determine the required power for user $o_{M-2}$ to user $o_1$, which are served by waveguide $n$. Similarly, the transmit power for all other users can also be calculated. For each waveguide, based on all its served users' transmit power, it is straightforward to obtain $P_n$ and $\balpha_n$. However, this is not the final step. The update of the transmission power for other users with different waveguides will affect the decoding order of the NOMA users within this waveguide. This, in turn, impacts the power allocation among the NOMA users. Therefore, we need to repeat the previous steps iteratively until no user's transmission power changes anymore. For brevity, we summarize the proposed power allocation method in Algorithm 1.
\begin{algorithm}[ht]
\caption{Proposed Power Control Algorithm}
\begin{algorithmic}[1]
\REQUIRE An initial transmit power of each user, i.e., $p_{n,m}$, for $n\in\cN$ and $m\in\cM$.
\STATE Set $t=1$ and the maximum iteration number as $T$.
\WHILE{$t\leq T$}
\STATE For each user, determine its nSINR based on \eqref{nsinr}, using the transmit power of the other users from the previous iteration, i.e., iteration $t-1$.
\STATE For each waveguide, determine the optimal decoding order for its associated users by arranging them in ascending order of their nSINRs.
\STATE Calculate the minimum required transmit power for users from user $o_M$ to user 1 based on the method introduced in Section IV-A.
\STATE {Update $t=t+1$.}
\ENDWHILE
\RETURN  $\balpha$ and $\bP$
\end{algorithmic}
\end{algorithm}
\subsection{Property Analysis}
For brevity, we define $\bP_{\bar o_m}$ as the transmit power of all users except user $o_m$ with waveguide $n$. It is evident that $p_{n,o_{m}}$ is a function of $\bP_{\bar o_m}$, denoted as $p_{n,o_{m}}=f(\bP_{\bar o_m})$, which we termed as the interference function of this user. Based on the discussion in Section IV-A, we know that for user $o_m$ with waveguide $n$, the required minimum transmit power is given by
\begin{equation} \label{f}
\begin{split}
    p_{n,o_{m}}&=f(\bP_{\bar o_m})\\
&=(2^{\frac{\bar R_{n,o_{m}}}{W}}-1)(\sum_{j=m+1}^Mp_{n,o_j} +\text{IIN}_{n,o_{m}}),
\end{split}
\end{equation}
where $\text{IIN}_{n,o_{m}}$  is related to the transmit power of all other users associated with the remaining $N-1$ dielectric waveguides, as depicted in \eqref{IN}.

With the aforementioned definitions, we summarize the properties of the developed scheme in the following theorem.
\begin{myth} \label{myth}
	The proposed power allocation Algorithm 1 converges to the unique fixed point for problem \eqref{obj}.
\end{myth}

\begin{proof}
According to \eqref{f}, we can check that the interference function $f$ satisfies the following three properties: (1) $f(\bP_{\bar o_m}) \geq 0$; (2) if $\bP_{\bar o_m} \geq \bP'_{\bar o_m}$, then $f(\bP_{\bar o_m}) \geq f(\bP'_{\bar o_m})$; and (3) for all $\beta >0$, $\beta f(\bP_{\bar o_m}) \geq f(\beta\bP_{\bar o_m})$. If a function satisfies the three criteria, we say that it is standard. Based on \cite{yate}, the iterative algorithm using a standard function
converges to the unique fixed point. This completes the proof.
\end{proof}
\begin{myprop}
   Based on Section IV-A, for the power allocation for users associated with each waveguide, only the normalized inter-waveguide interference plus noise is required. The channel gain and specific power control methods for other users are not necessary. Thereby, the developed power allocation method is distributed and applicable to practical time-sensitive wireless networks.
\end{myprop}

In the next section, we will use extensive numerical results to show that the developed algorithm can converge within several iterations.

\section{Numerical Results}
In this section, we present numerical simulation results to validate the performance of the proposed power control method for a NOMA-pinching antenna system with multiple waveguides. For brevity, the default values of the simulation parameters, mainly following the settings in \cite{ding}, are summarized in Table \ref{DefaultValues}. Moreover, we use the labels ``Pro.'' and ``Equ.'' to denote the proposed power control algorithm and the equal power allocation strategy, respectively.
\begin{table}[ht]
\small
\centering
\renewcommand{\arraystretch}{1.1}
\caption{Default Values of Simulation Parameters}
\begin{tabular}{|>{\centering\arraybackslash}m{5.5cm} | >{\centering\arraybackslash}m{2cm}|}
\hline
\textbf{Parameter} & \textbf{Values} \\
\hline
Carrier frequency $f_c$ & $28$ GHz\\
\hline
Height of the antenna $d$ & $3$ m\\
\hline
Number of waveguides $N$ & $1$-$7$ \\
\hline
Number of users $M$ & $2$-$3$ \\
\hline
Waveguide wavelength $\lambda_g$ in a dielectric waveguide & $\lambda/1.4$\\
\hline
Noise power $\sigma^2$ & $-90$ dBm\\
\hline
Bandwidth $W$ & $10$ MHz\\
\hline
Waveguide's feed point location $\psi_{n,0}$ & $(0,0,d)$\\
\hline
Positions of pinches $\psi_{n,i}$ & $(iD,nD,d)$\\
\hline
Positions of users $\tilde \psi_{n,m}$ & $(mD,nD,0)$\\
\hline
Interval of the waveguides and pinches $D$ & $20$ m\\
\hline
\end{tabular}
\label{DefaultValues}
\end{table}

Fig. \ref{fig_interation} demonstrates the convergence performance of the proposed power control algorithm for different configurations with the number of waveguides $N=2$ and $N=3$, as well as the transmission rates constraints $\bar R_{n,m}$ of $5$ Mbps, $10$ Mbps, and $15$ Mbps. Given the same initial conditions, the proposed algorithm successfully converges within merely ten iteration cycles, which demonstrates high convergence efficiency. Nevertheless, the convergence efficiency of our proposed algorithm is influenced by the transmission rates. Higher transmission rates necessitate higher power levels, leading to increased interference power during iteration, which consequently reduces the convergence efficiency. Furthermore, adding a waveguide to meet the transmission rates constraints of $5$ Mbps, $10$ Mbps, and $15$ Mbps requires additional transmit power level ratios of $51.5\%$, $57.3\%$, and $94.0\%$ respectively, which aligns with a $50\%$ increase in users and sum rate of the system. However, it is evident that the transmit power ratios increase with higher transmission rate constraints. As discussed previously, more stringent requirements of the transmission rate inevitably result in increased power interference among the users on different waveguides, which consequently compromises the system's performance. Hence, the interval of the waveguides and pinches should be adequately maintained to achieve low power consumption under high transmission rate constraints.
\begin{figure}
    \centering
    \includegraphics[width=0.73\linewidth]{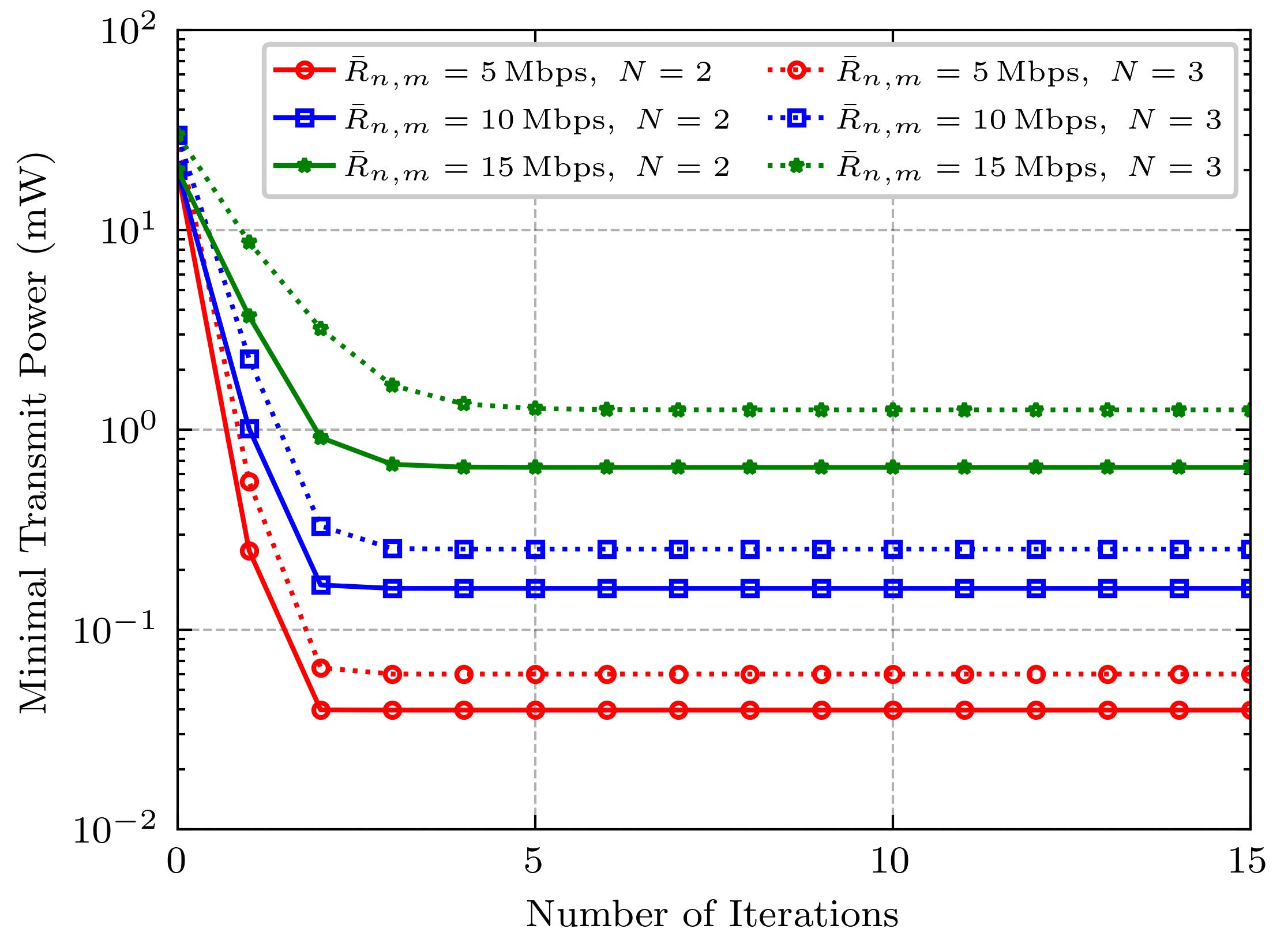}
    \caption{The minimal transmit power versus the number of iterations}
    \label{fig_interation}
\end{figure}

Fig. \ref{fig_waveguide} compares the minimal transmit power between the proposed power control algorithm and the equal power allocation algorithm. Taking the case of $N=2$ and $R=10$ Mbps as an example, our proposed algorithm reduces the minimal transmit power by $76.0\%$ and $93.2\%$ for $M=2$ in Fig. \ref{fig_waveguide:a} and $M=3$ in Fig. \ref{fig_waveguide:b}, respectively, compared to the equal power allocation algorithm. The better performance of our proposed algorithm results from the fact that the upper bound of the transmission rate $\bar R_{n,o_m}\le R_{n,o_m}\le W\log_2(1+\frac{1}{M-o_m})$ is unavoidable in the equal power allocation algorithm, which can be proved through \eqref{rate} and \eqref{sinr}. Consequently, under high transmission rate constraints, the equal power strategy obtains the suboptimal solution that are either excessively large or computationally infeasible to obtain, which can be observed from the case of $\bar R_{n,m}=15$ Mbps in Fig. \ref{fig_waveguide:b}. Furthermore, it is also illustrated in Fig. \ref{fig_waveguide:a} that the minimal transmit power increases with the number of waveguides. Specifically, in systems with few waveguides, adding new waveguides requires significantly more power, whereas the pinching antennas system with more waveguides demands nearly equal power when adding new waveguides, since the slope remains almost identical over the interval $N\in[4,7]$ in Fig. \ref{fig_waveguide:a}.
\begin{figure}[htbp]
    \centering
    \subfloat[$M=2$]{%
        \includegraphics[width=0.73\linewidth]{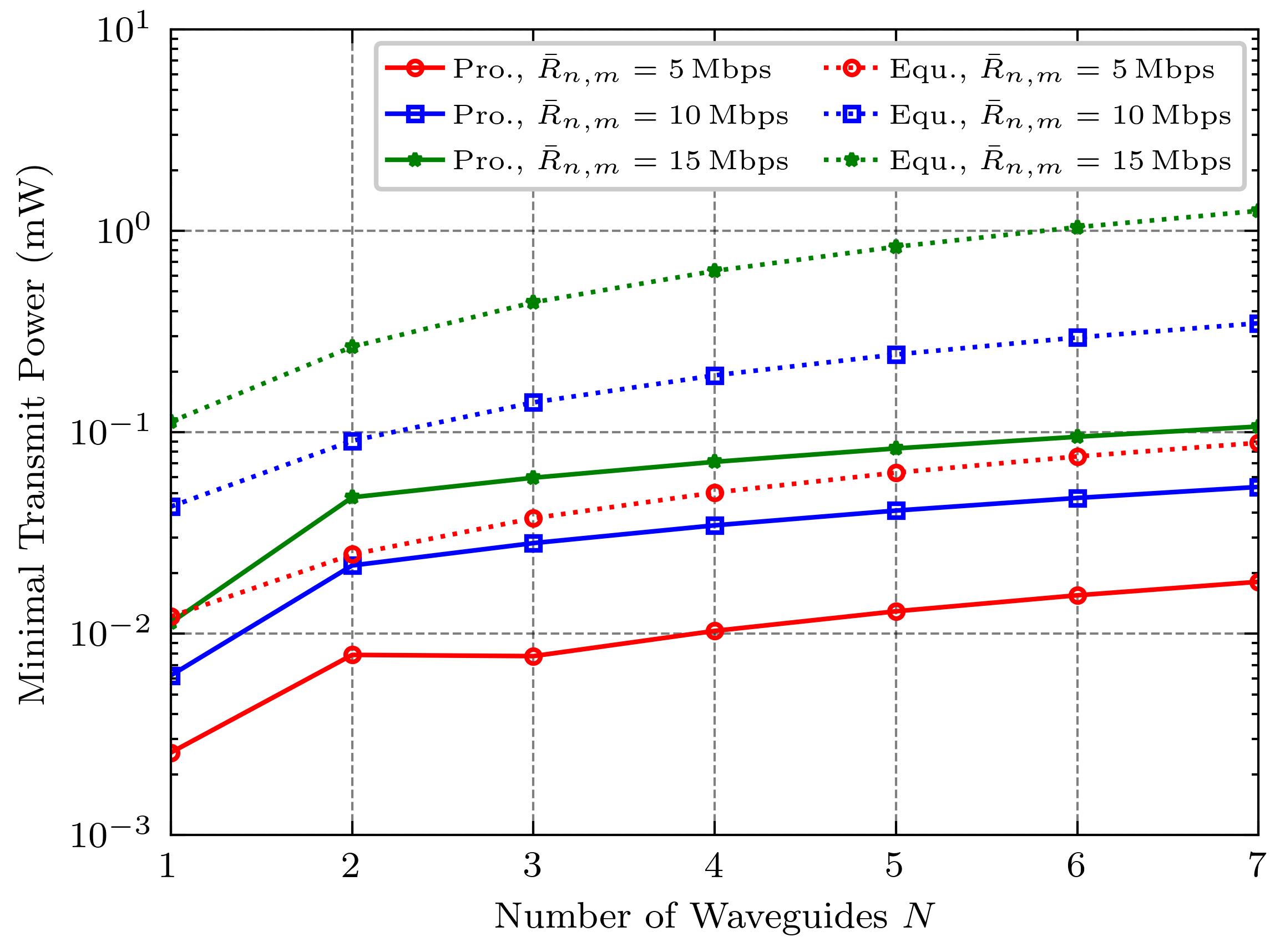}%
        \label{fig_waveguide:a}%
    }
    \hfill
    \subfloat[$M=3$]{%
        \includegraphics[width=0.73\linewidth]{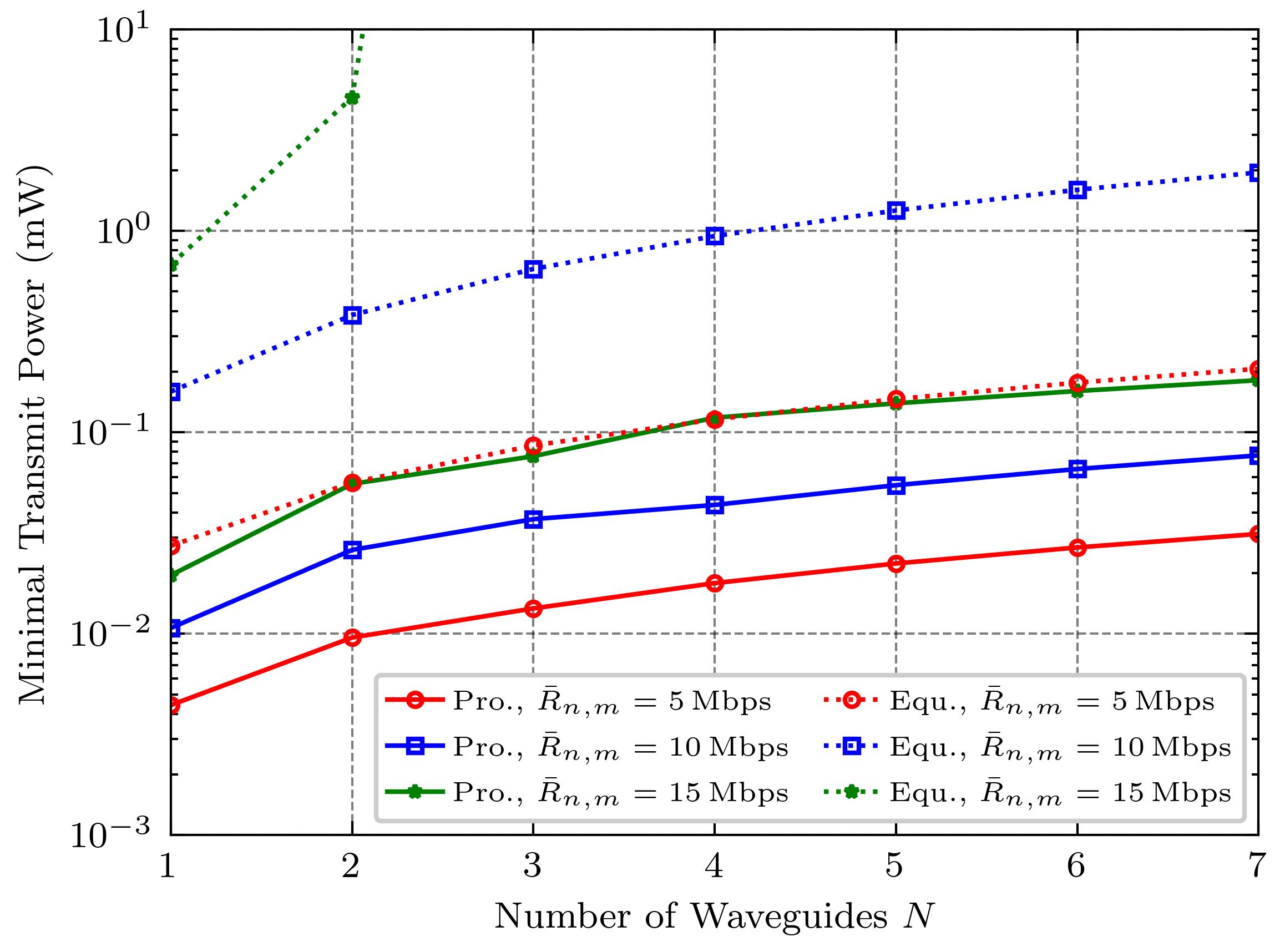}%
        \label{fig_waveguide:b}%
    }
    \caption{The minimal transmit power versus the number of waveguides.}
    \label{fig_waveguide}
\end{figure}
\begin{figure}[htbp]
    \centering
    \subfloat[$N=2$]{%
        \includegraphics[width=0.49\linewidth]{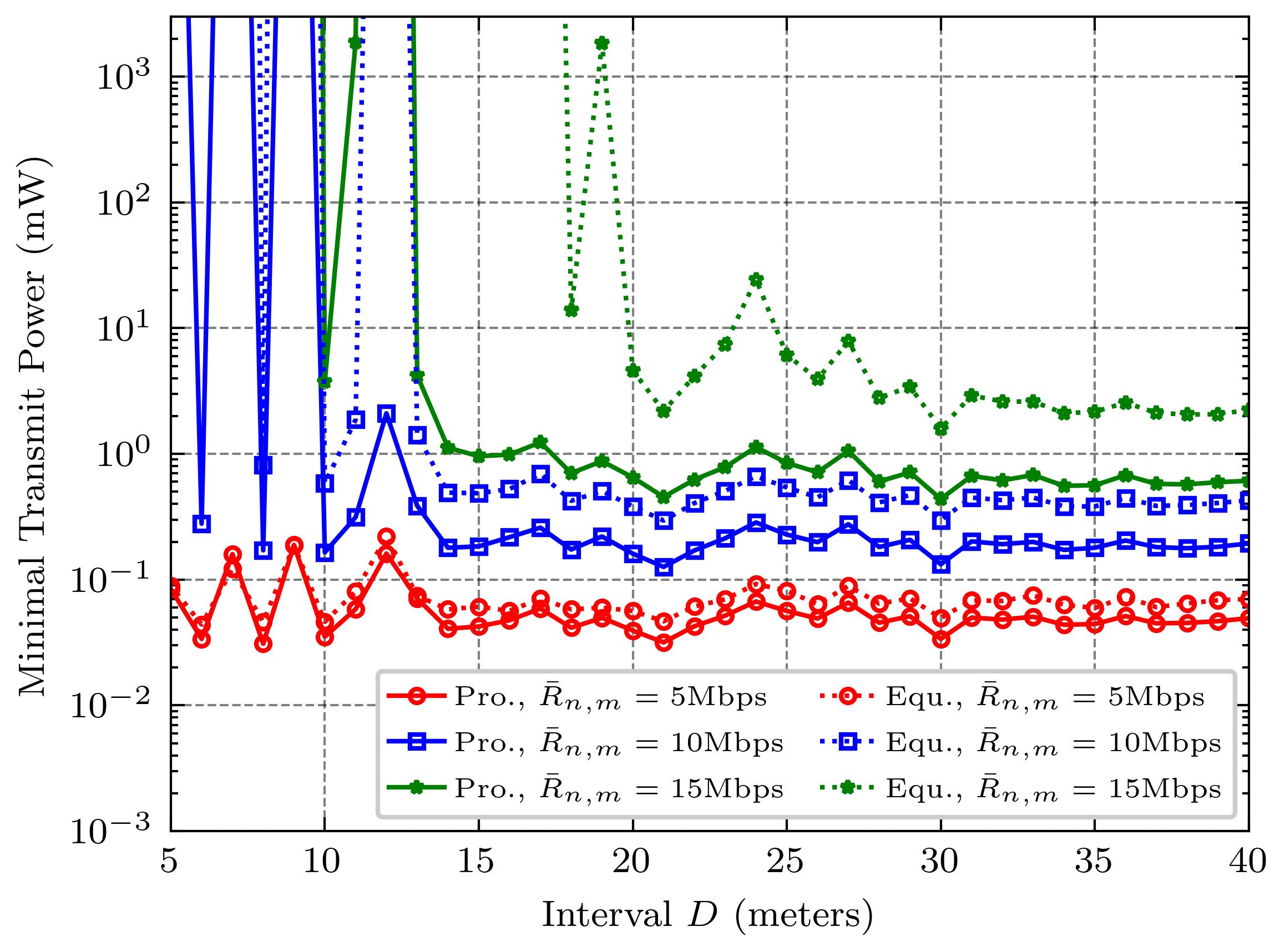}%
        \label{fig_interval:a}%
    }
    \subfloat[$N=3$]{%
        \includegraphics[width=0.49\linewidth]{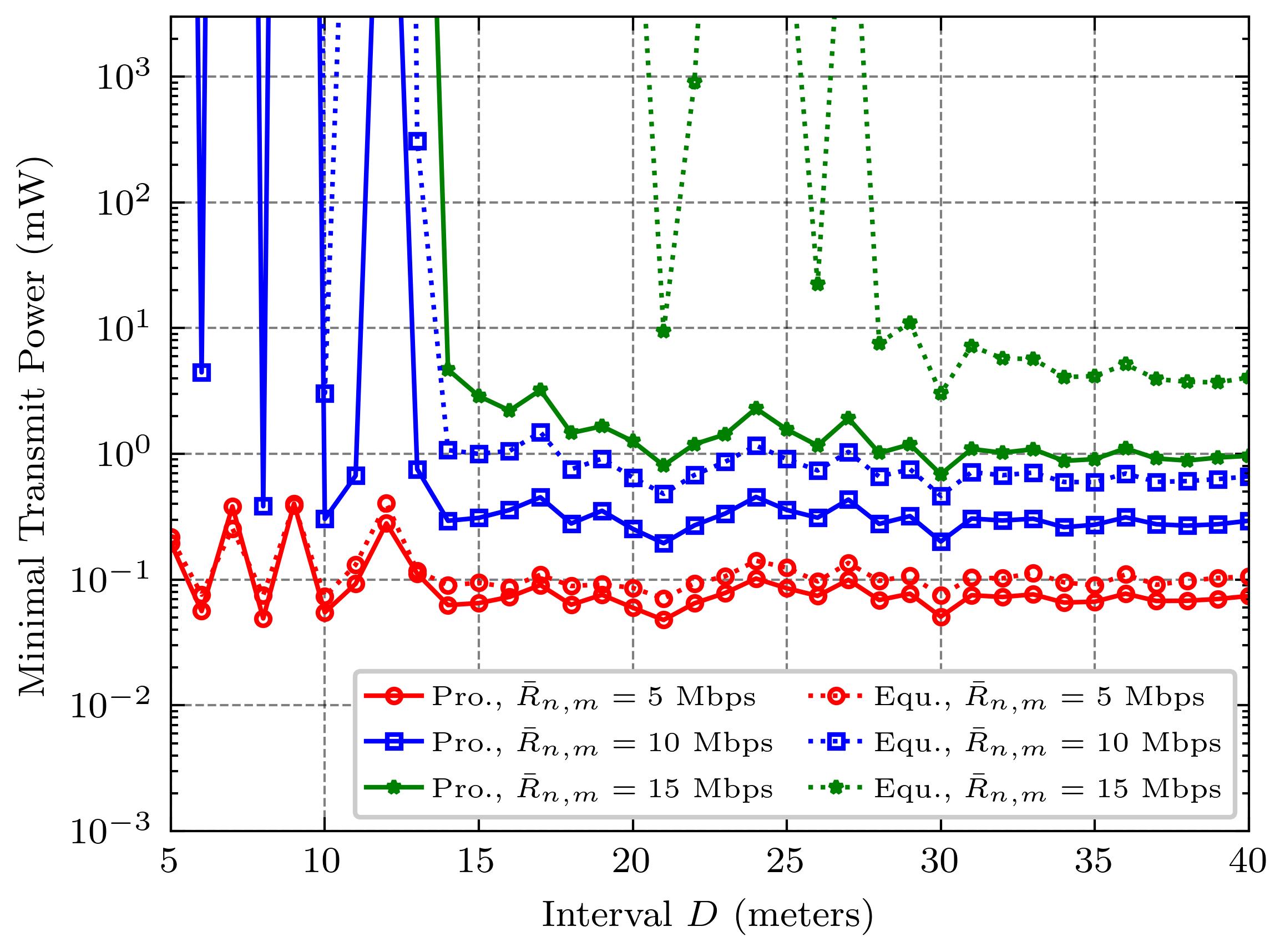}%
        \label{fig_interval:b}%
    }
    \hfill
    \subfloat[$\text{Zoom-in plot for } N=2$]{%
        \includegraphics[width=0.49\linewidth]{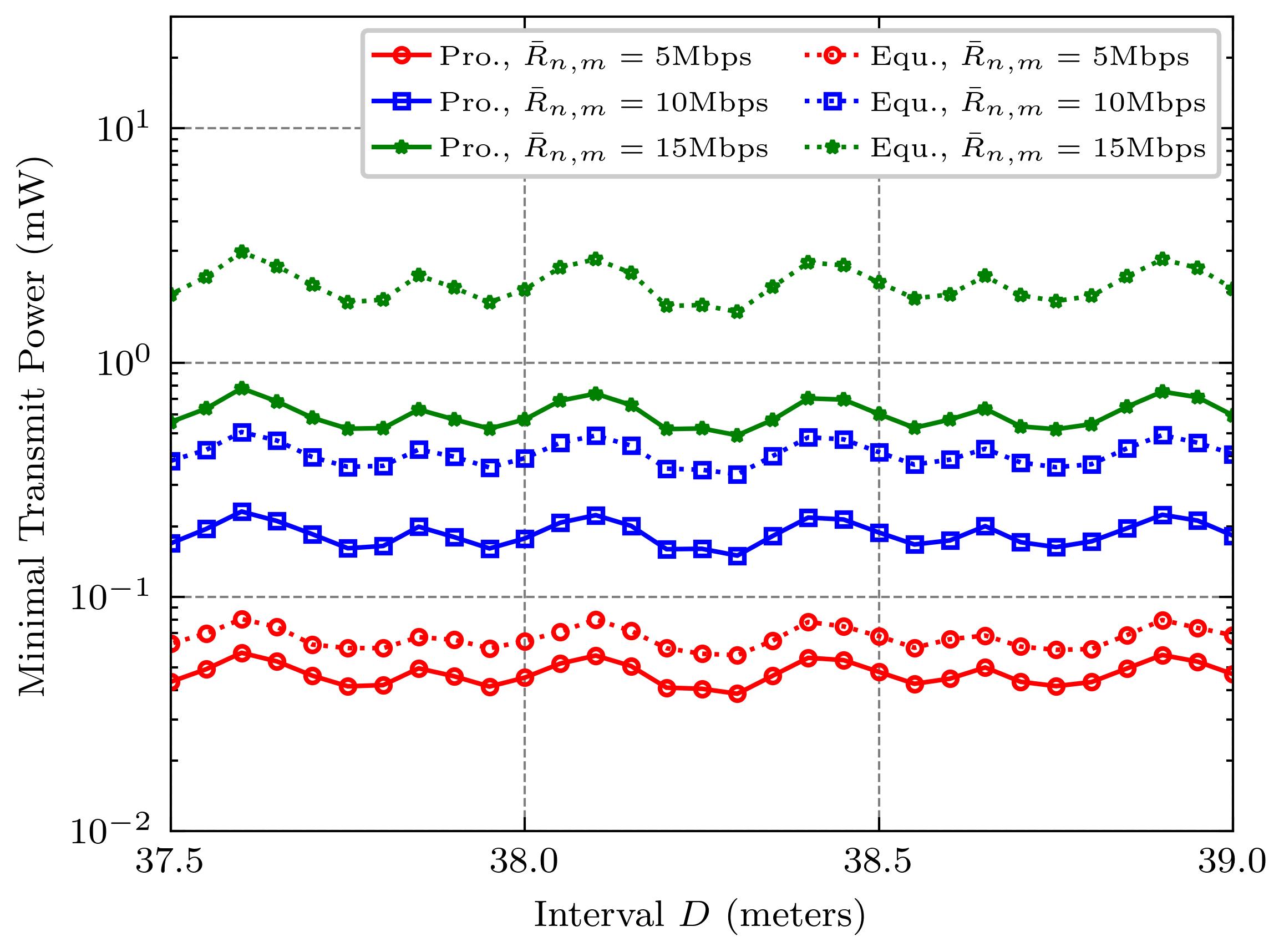}%
        \label{fig_interval_minor:a}%
    }
    \subfloat[$\text{Zoom-in plot for } N=3$]{%
        \includegraphics[width=0.49\linewidth]{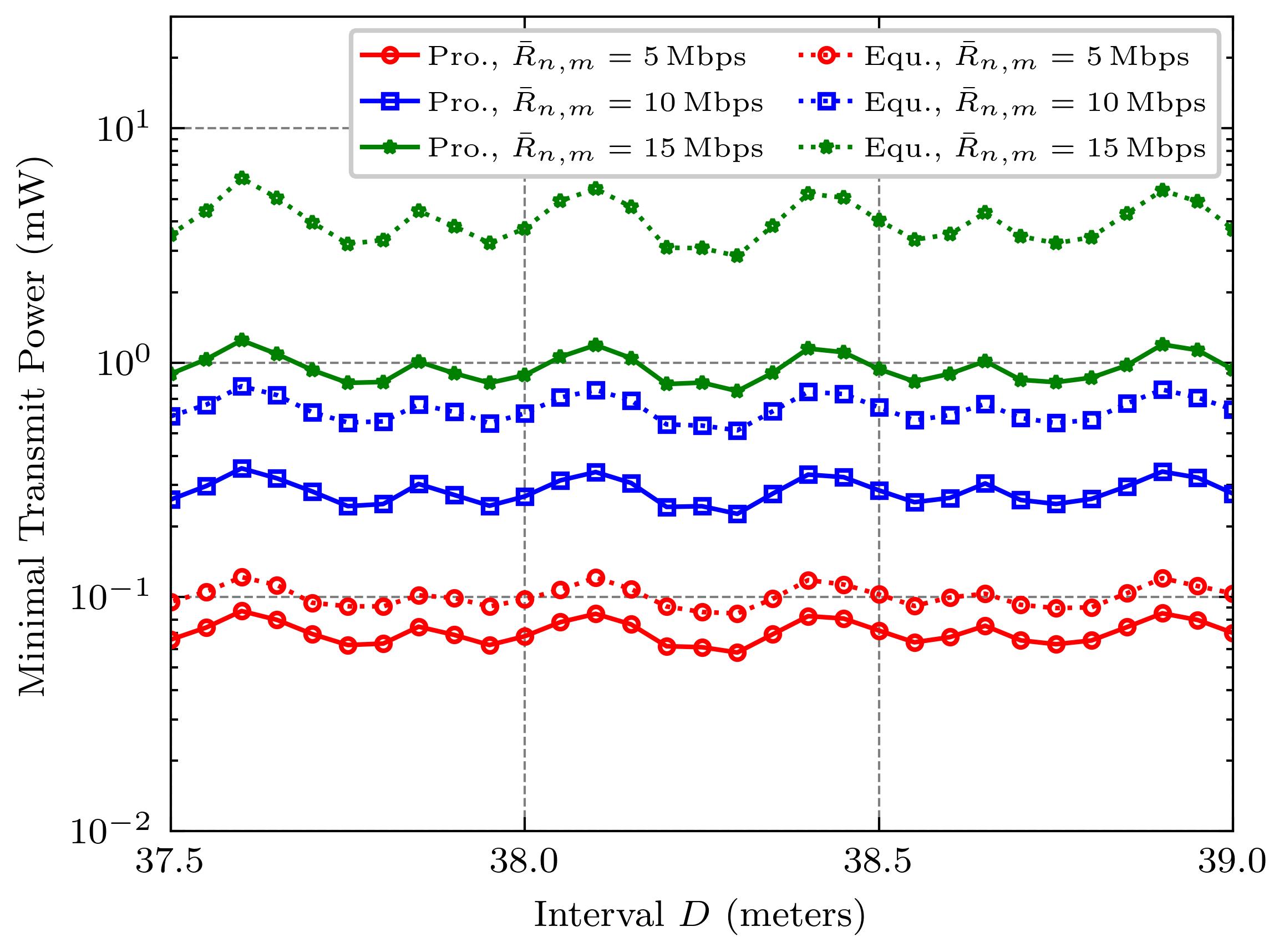}%
        \label{fig_interval_minor:b}%
    }
    \caption{The minimal transmit power versus the interval $D$.}
    \label{fig_interval}
\end{figure}

Fig. \ref{fig_interval} is displayed to reveal the relationship between the minimal transmit power and the interval of the waveguides and pinches. It is demonstrated that the minimal transmit power, as a function of the interval $D$, displays an \textit{approximate oscillatory decay} with a negative exponential envelope. The observed functional behavior motivates us to approximate it with the underdamped sine function in relation to the interval $D$. More specifically, within the long ranges of $D\in[1,40]$, the function exhibits decay with a negative exponential envelope in Fig. \ref{fig_interval:a} and Fig. \ref{fig_interval:b}, while within the short ranges of $D\in[37.5,39]$, the function demonstrates approximate periodicity in Fig. \ref{fig_interval_minor:a} and Fig. \ref{fig_interval_minor:b}. The fundamental reason lies in that the definition of channel gain $h_{m,n}$ in \eqref{channel_gain} can be represented by using the underdamped sine function. Therefore, the properties of the $\text{SINR}_{n,o_m}$ in \eqref{sinr}, even and the minimal transmit power, share similarities with the underdamped sine function. To precisely reveal the insights in Fig. \ref{fig_interval}, we take the case of $R=10$ Mbps in Fig. \ref{fig_interval:b} as an example. Specifically, each lower value of the minimal transmit power, i.e., $D=6$ meters, essentially indicates the better the channel condition, where the channel gain magnitude $|h_{n,m}|$ of the current waveguide significantly exceeds that of other waveguides. As the interval $D$ increases, the minimal transmit power approaches a constant. Since the increment $|\tilde \psi_{n,m}- \psi_{n,i}|=\sqrt{(m-i)^2D^2+d^2}$ increases, thereby leading to a corresponding decrease in the channel gain magnitude $|h_{n,m}|$. Finally, more power must be allocated to meet the constraints of the transmission rates. Besides, given the large interval $D$, it follows that the increment $|\tilde \psi_{n,m}- \psi_{n,i}|=\sqrt{(m-i)^2D^2+d^2}$ is the much larger coefficient for $\forall m \neq i$, compared to the increment $|\tilde \psi_{n,m}- \psi_{n,m}|=d$. Hence, the channel gain $h_{n,m}$ associated with user $m$ on waveguide $n$ is nearly independent of pinches far from the user. In other words, the $N$ correlated waveguides are decoupled into $N$ independent waveguides, and each channel gain magnitude $|h_{n,m}|$ is approximately expressed as $|h_{n,m}|\approx \sqrt{\eta}/d$. Evidently, the minimal transmit power of the NOMA-assisted pinching antenna systems tends to a constant as $D\to +\infty$. Besides, the decay rate and the degree of oscillation in Fig. \ref{fig_interval} are influenced by the transmission rate constraints, which implies the tradeoff between transmission rate and NOMA-assisted pinching antenna system stability. As the transmission rate increases, the interference among users on other waveguides undoubtedly becomes more severe. Therefore, the design of the interval $D$ plays a crucial role in determining the minimal transmit power allocation of the system.


\section{Conclusion}
This paper studied a classical power minimization problem for NOMA-assisted pinching antenna systems. In this setup, each dielectric waveguide is equipped with multiple pinching antennas that serve multiple NOMA users. We formulated a power allocation problem while ensuring that each user meets a minimum data rate requirement. To efficiently solve this non-convex optimization problem, we developed an iterative algorithm with a standard interference function that converges within a few steps. Extensive simulation results demonstrate the efficiency and superior power-saving capability of our approach compared to benchmark schemes. In future work, we plan to jointly optimize user grouping, antenna positioning, and power allocation for NOMA-aware pinching antenna networks, as these factors collectively influence the performance of NOMA-pinching antenna systems.

\section*{Acknowledgment}
The system model, problem formulation, and methodology design were contributed by Dr. Yaru Fu, while the simulation results and analysis were contributed by Mr. Fuchao He. Prof. Zheng Shi and Prof. Haijun Zhang assisted with proofreading.

\bibliographystyle{IEEEtran}
\bibliography{reference}
\end{document}